\documentclass[10pt]{ijnam}
\def\currentvolume{1}
\def\currentissue{1}
\def\currentyear{2008}
\def\month{January}
\pagespan{1}{20}
\copyrightinfo{2010}{} 

\usepackage{amssymb,amsmath}
\usepackage{amsmath,amssymb,latexsym,amsfonts}
\usepackage{amsthm}
\usepackage[dvips]{graphicx}

\newtheorem{thm}{Theorem}

\newtheorem{remark}{Remark}

\newcommand{\R}{\mathbb{R}}

\newcommand{\E}{\varepsilon}

\newcommand{\D}{^{\prime}}
\newcommand{\DD}{^{\prime\prime}}

\begin{document}

\title{
Weakly nonlinear analysis of the Hamilton--Jacobi--Bellman equation arising from pension savings management}%

\author{Zuzana Macov\'a \and Daniel \v{S}ev\v{c}ovi\v{c} }

\address{
Dept. of Applied Mathematics and Statistics, 
 Faculty of Mathematics, Physics and Informatics, Comenius University, 
 842 48 Bratislava, Slovakia
}

\email{sevcovic@fmph.uniba.sk}

%
\commby{Lubin G. Vulkov}

\date{April 30, 2009}


\abstract{
The main purpose of this paper is to analyze solutions to a fully nonlinear parabolic equation arising from the problem of optimal portfolio construction. We show how the problem of optimal stock to bond proportion in the management of pension fund portfolio can be formulated in terms of the solution to the Hamilton--Jacobi--Bellman equation. We analyze the solution from qualitative as well as quantitative point of view. We construct useful bounds of solution yielding estimates for the optimal value of the stock to bond proportion in the portfolio. Furthermore we construct asymptotic expansions of a solution in terms of a small model parameter. Finally, we perform sensitivity analysis of the optimal solution with respect to various model parameters and compare analytical results of this paper with the corresponding known results arising from time-discrete dynamic stochastic optimization model.}


\maketitle
\pagestyle{myheadings}
\markboth{Z. Macov\'a \and D. \v{S}ev\v{c}ovi\v{c}}{Weakly nonlinear analysis of the Hamilton--Jacobi--Bellman equation}

\section{Introduction and problem formulation}

In this paper we are analyzing solutions to the Hamilton--Jacobi--Bellman equation arising from stochastic dynamic programming for optimal decision between stock and bond investments during accumulation of pension savings. 
Such an optimization problem often arises in optimal dynamic portfolio
selection and asset allocation policy for an investor who is concerned about the performance of a portfolio relative to the performance of a given benchmark (see e.g. \cite{MERTON1971,MERTON1974,MERTON1975,MERTON1992,BROWNE1995,BODIE1995,BODIE2003,KILIANOVA2009,KILIANOVA2009b}).

Consider the function $V(t,y), (t,y)\in \mathcal{D}$, defined on a domain 
$\mathcal{D} = [0,T)\times (0,\infty)$
and satisfying the following  fully nonlinear Hamilton--Jacobi--Bellman parabolic partial differential equation:
\begin{subequations}
\begin{equation}\label{HJBE}
\frac{\partial V}{\partial t}
+
\max_{\theta \in \Delta_t}
\left(
A_{\E}(\theta,t,y)\frac{\partial V}{\partial y}+ \frac{1}{2}B^2(\theta,t,y)\frac{\partial^2 V}{\partial y^2}
\right) = 0, \qquad (t,y)\in \mathcal{D},
\end{equation}
and the terminal condition at $t=T$, 
\begin{equation}\label{TC}
    V(T,y) = U(y), \quad y\in (0,\infty),  
\end{equation}
\end{subequations}
where $U = U(d)$ is a smooth strictly increasing concave bounded function 
and $\E$ is a small parameter, $0 <\E\ll 1 $. 
Moreover, we suppose that the following additional requirements are met:
\begin{enumerate}
    \item\label{A1} the admissible  set $\Delta_t = [l_t, u_t] \subset\R$ for all $0 \leq t\leq T$;
    \item\label{A2} the function $\Delta_t\ni\theta\mapsto A_{\E}(\theta,t,y)\in\R$ is (not necessarily strictly) concave in the $\theta$ variable and it is function increasing at $\theta=l_t$; 
    \item\label{A3} the function $\Delta_t\ni\theta\mapsto B^2(\theta,t,y)$ is strictly convex in the $\theta$ variable and it is decreasing at $\theta=l_t$. 
\end{enumerate}

Let us suppose for a moment that the function $y\mapsto V(t,y)$ is an increasing and strictly concave function in the $y$ variable. Then applying the first order necessary condition on the maximum of the function 
\[
\theta\mapsto A_{\E}(\theta,t,y)\frac{\partial V}{\partial y}+ \frac{1}{2}B^2(\theta,t,y)\frac{\partial^2 V}{\partial y^2}
\]
we obtain the following implicit equation for $\hat{\theta}$, the maximizer of the above function:
\begin{equation}\label{IE}
    G(\hat{\theta},t,y) = - \frac{\frac{\partial V}{\partial y}(t,y)}{y \frac{\partial^2V}{\partial y^2}(t,y)}
\quad \hbox{where}\ \ 
G(\theta,t,y) =\frac{1}{2}\frac{\frac{\partial (B^2)}{\partial \theta}}{y\frac{\partial A_{\E}}{\partial \theta}}.
\end{equation}
Since the requirements \eqref{A2}--\eqref{A3} guarantee the increase of the function $G(\theta, t,y)$ in the $\theta$ variable, there exists the inverse of $G$ and thus the unique $\hat{\theta}=\tilde{\theta}(t,y)$ such that
\[
\hat{\theta}(t,y) = 
G^{-1}\left(  - \left(\frac{\partial V}{\partial y}(t,y) \right)/\left(y\frac{\partial^2V}{\partial y^2}(t,y)\right)\right).
\]
Then the optimal value of $\theta$ solving \eqref{HJBE} with the terminal condition \eqref{TC} is given by
\begin{equation}\label{OptTheta}
 \theta^{\ast}(t,y) = 
\min\{u_t,\hat{\theta}(t,y)\}.
\end{equation}
The problem \eqref{HJBE} can be now treated as a fully nonlinear parabolic partial differential equation of the form:
\begin{subequations}
\begin{equation}\label{HJBEopt}
    \frac{\partial V}{\partial t}+\mathcal{F}(t,y,V,\frac{\partial V}{\partial y}, \frac{\partial^2 V}{\partial y^2}) = 0
\end{equation}
where
\begin{equation}
\mathcal{F}(t,y,V,\frac{\partial V}{\partial y}, \frac{\partial^2 V}{\partial y^2}) 
= A_{\E}(\theta^\ast(t,y),t,y) \frac{\partial V}{\partial y}
+
\frac{1}{2}B^2(\theta^\ast(t,y),t,y)  \frac{\partial^2 V}{\partial y^2}
\end{equation}
\end{subequations}
where $\theta^\ast$ is given by (\ref{OptTheta}) and $\hat\theta(t,y) 
= G^{-1}\left(  - \left(\frac{\partial V}{\partial y}(t,y) \right)/\left(y\frac{\partial^2V}{\partial y^2}(t,y)\right)\right)$ depends itself on the solution $V$ and its derivatives. The solution is subject to the terminal condition $V(T,y) = U(y)$ where  $V = V(t,y)$ for $y>0$ and $0 \leq t \leq T$. Furthermore, $\frac{\partial \mathcal{F}}{\partial q}>0$. 

\medskip
The application of this study to financial markets, particularly to the theory of the optimal portfolio construction, has a strong impact on the special choice of the functions $A_{\E}$ and $B$ used in the original formulation of the studied problem \eqref{HJBE}. Hence let us consider
\begin{equation}\label{SC}
    A_{\E}(\theta,t,y) = \E+[\mu_t(\theta) - \beta_t]y,  \qquad \qquad B(\theta,t,y) = \sigma_t(\theta)y, 
\end{equation}where $\{\beta_t; 0 \leq t \leq T\}$, $\{\mu_t(\theta); 0 \leq t \leq T\}$ and $\{\sigma_t(\theta); 0 \leq t \leq T\}$ are assumed to be given deterministic processes for any choice of the control parameter $\theta \in \Delta_t$. 

Moreover, if $V(t,y)$ is strictly convex and increasing in the $y$ variable, then with regard to assumptions \eqref{A2}--\eqref{A3} the monotonicity of the function  $\theta\mapsto G(\theta,t,y)$ is guaranteed. Thus the unique maximizer $\tilde{\theta}$ rising up from the implicit equation \eqref{IE} fulfills
\begin{equation}\label{STopt}
    \tilde{\theta} (t,y)= G^{-1}\left(  - \frac{\frac{\partial V}{\partial y}(t,y)}{y\frac{\partial^2V}{\partial y^2}(t,y)}\right) 
\quad \hbox{where} \ 
 G(\theta,t)= \frac{1}{2}\frac{[\sigma_t^2(\theta)]\D}{\mu_t(\theta)\D}.
\end{equation}
Here $(.)^\prime$ denotes the derivative with respect to the $\theta$ variable and the inverse $G^{-1}$ is taken with respect to the $\theta$ variable.

The paper is organized as follows. In the next section we present a stochastic dynamic optimization model that can be used for solving the dynamic problem of optimal stock to bond proportion when managing the saver's pension fund portfolio. We recall key steps of derivation of a time discrete version of the model. Then we propose a time continuous version of the stochastic dynamic optimization problem. We take as our setting the standard continuous-time framework pioneered by Merton, Samuelson, Bodie, Browne and others (c.f. \cite{MERTON1971,MERTON1974,MERTON1975,MERTON1992,BROWNE1995,BODIE1995,BODIE2003,SAM1969,KILIANOVA2009}). It will be shown that the intermediate value (or utility) function satisfies certain fully nonlinear parabolic equation of Hamilton--Jacobi--Bellman type. In Section 3, following ideas of the recent paper by Abe and Ishimura \cite{ISHIMURA2008}, we introduce a Riccati like transformation in order to transform the fully nonlinear Hamilton--Jacobi--Bellman equation into the quasilinear parabolic equation (see also \cite{ISHIMURA2009}). Then we derive some useful bounds of a solution by using a parabolic comparison principle. These parabolic estimates are also used in order to provide bounds for the original variables -- both the intermediate utility function and for the optimal stock to bond proportion. A solution to the transformed quasilinear parabolic equation is constructed by means of weakly nonlinear analysis, i.e. we seek a solution in the form of infinite power series with respect to a small parameter representing yearly percentage of a salary transfered to saver's pension account. We show that first three terms in the expansion can be explicitly found. We also provide a recurrent formula for calculating higher order terms in the expansion. The rest of Section 3 is devoted to the sensitivity analysis of the optimal stock to bond proportion with respect to various model parameters, in particular, to the yearly percentage transfer and to the saver's risk aversion. In Section 4 we demonstrate results of numerical simulations with model parameters corresponding to the second pillar of the pension system in Slovakia. We compare the results obtained with help of explicit approximations of an optimal solution to the time continuous Hamilton--Jacobi--Bellman equation with those obtained by the time discrete model from  Melicher\v{c}\'{\i}k {\it et al.}  \cite{KILIANOVA2006,MS2009}

\section{Motivation and model derivation}
In this section we first  recall a discrete dynamic stochastic optimization
problem arising in optimal portfolio selection. The discrete version of this model has been derived
by Kilianov\'a, Melicher\v{c}\'{\i}k and the author in \cite{KILIANOVA2006}. It was applied for
solving a problem of construction of an optimal stock to bond proportion in pension fund selection 
for the second pillar of the Slovak pension system. In what follows, we
recall key steps in derivation of the discrete dynamic stochastic
optimization pension savings model due to  Melicher\v{c}\'{\i}k {\it et al.}
in \cite{KILIANOVA2006}. In the second part of this section we shall 
generalize the model from its discrete version to a
continuous one. It will be shown that the continuous model for solving a problem of 
optimal stock to bond proportion in pension fund selection can be
reformulated in terms of a fully nonlinear parabolic equation also referred to as the Hamilton--Jacobi--Bellman equation.

In the discrete optimal pension fund selection model due to Melicher\v{c}\'{\i}k {\it et al.} \cite{KILIANOVA2006,MS2009}, a future pensioner with the expected retirement time in $T$ years transfers regularly once a year an $\E$-part of his yearly salary with the deterministic rate of growth $\beta_t$ to the pension fund investing in financial market with the yearly stochastic return $r_t$. More precisely, we denote  by $B_t$ his yearly salary at the year $t$. Then the budget constraint equation for the total accumulated sum $Y_t$ in his pensioner's account reads as follows:
\[
Y_{t+1} = (1+ r_t) Y_t + \E B_{t+1},\ \ \hbox{for}\ t=1,2, ..., T-1,\qquad Y_1 = \E B_1.
\]
Supposing the wage growth $\beta_t$ is known, we have the relation $B_{t+1} = (1+\beta_t) B_t$ between two consecutive yearly salaries. At the time of retiring a future pensioner will aim at maintaining  his living standards compared to the level of the last salary at the retirement time $t=T$. 
Therefore the absolute value of the total saved sum $Y_T$ at the time of retirement $T$ does not represent the quantity a future pensioner will be taking care about. More important information for him is expressed by a ratio of the cumulative saved sum $Y_T$ and the yearly salary $B_T$, i.e. $y_t=Y_t/B_t$ at $t=T$. In terms of the quantity $y_t$ representing the number of yearly salaries already saved at time $t$, the budget-constraint equation can be reformulated as follows:
\[
y_{t+1} = y_t (1+ r_t)(1+\beta_t)^{-1}  + \E ,\ \ \hbox{for}\ t=1,2, ..., T-1,\qquad y_1 = \E.
\]

For the sake of simplicity, we assume that the investment strategy of the pension fund at time $t$ is given by the proportion $\theta\in[0,1]$ of stocks and $1-\theta$ of bonds and that the fund return $r_t$ is normally distributed 
with the mean value $\mu_t(\theta)$ and dispersion $\sigma_t^2(\theta)$ for any choice of the stock to bond proportion $\theta$. It means that
\begin{equation} 
\label{rNormal} r_t(\theta) \sim N(\mu_t(\theta), \sigma^2_t(\theta)), \quad i.e.\quad 
r_t(\theta) = \mu_t(\theta) + \sigma_t(\theta) Z
\end{equation}
where $Z\sim N(0,1)$ is a normally distributed random variable having the probabilistic density function $f(z)=\frac{1}{\sqrt{2\pi}} \exp(-z^2/2)$.
Both $\mu_t$ and $\sigma_t^2$ depend directly on the choice of parameter $\theta$ 
representing stock to bond proportion in the portfolio. It assumed to belong 
to the prescribed admissible set $\Delta_t = [l_t, u_t] \subseteq [0,1]$ for any time 
$t \in [0,T]$. The admissible set $\Delta_t$ is subject to governmental regulations
that may be imposed on the stock to bond proportion in a specific time $t \in [0,T]$.
At each time $t$, the mean value and volatility of the fund return $r_t$ can be expressed in terms of expected values of returns $\mu_t^{(s)}, \mu_t^{(b)}$ and volatilities $\sigma_t^{(s)}, \sigma_t^{(b)}$ of stocks and bonds as follows:
\begin{equation}\label{mu}
    \mu_t(\theta) = \theta \mu_t^{(s)}+(1-\theta)\mu_t^{(b)}, 
\ \ 
    \sigma_t^2(\theta)  = \theta^2[\sigma_t^{(s)}]^2+(1-\theta)^2[\sigma_t^{(b)}]^2+2\theta (1-\theta)
\sigma_t^{(s)}\sigma_t^{(b)}\varrho_t, 
\end{equation}
where $\varrho_t \in [-1,1]$ is a correlation coefficient between the returns
on stocks and bonds at time $t$ and the time-independent values of the parameters $\mu^{(s)}$, $\mu^{(b)}$, $\sigma^{(s)}$ and $\sigma^{(b)}$ are known at time $t \in [0,T]$, they follow their relevant mutually independent Markov processes.

Thus the time-evolution of the number of allocated yearly salaries can be formulated by the following recurrent equation
\begin{eqnarray} 
\label{DRE} 
y_{t+1} &=& G_t^1(y_t, r_t(\theta_t) ), \qquad y_1 = \varepsilon, \\ 
\hbox{where}\ \ G_t^1(y, r_t) &=&  \E+y\frac{1+r_t}{1+\beta_t} \quad  \text{for $t = 1,2,\ldots, T-1$}.\nonumber 
\end{eqnarray}
Notice that $r_t(\theta)$ is the only stochastic variable appearing in the recurrent definition of the processes for the amount $y_t$ of yearly saved salaries. 
Our aim is to determine the optimal strategy, i.e. the optimal value of the weight $\theta_t$ at each time $t$ that maximizes the contributor's utility from the terminal  wealth allocated on their pension account, and so taking into account knowledge of the saver's utility function $U$, the problem of discrete stochastic dynamic programming can be formulated as
\begin{equation}\label{DSP}
\max_{\mathcal{S}}\mathbb{E}(U(y_T))
\end{equation} 
subject to the constraint \eqref{DRE} where the maximum in the stochastic dynamic problem is taken over all non-anticipative strategies, time sequences of $\theta_t$ stocks proportions, $\mathcal{S} = \{(t,\theta_t)\;|\; t = 1,\ldots, T\}$. Therefore the optimal strategy of the problem \eqref{DSP} subject to \eqref{DRE} is the solution to the Bellman equation
\begin{equation}
W(t,y) = \left\{ 
\begin{matrix}
U(y), \hfill& t = T,\hfill \cr
\displaystyle{\max_{\theta \in\Delta_t}}\ \  \mathbb{E}_Z \left( W(t+1, F_t^1(\theta,y,Z) )\right),\hfill &  t=T-1, ..., 2,1,\hfill 
\end{matrix} 
\right.
\end{equation}
where $F_t^1(\theta,y,z) = G_t^1(y, \mu_t(\theta) + \sigma_t(\theta) z)$.

In this paper the major object of our study is the continuous version of the discrete model proposed above. Instead of time intervals $[t,t+1], t=1, ..., T-1,$ representing discrete saving periods we shall assume that the proportion of the size $\E \tau$ of saving deposits is transfered to the saver account on short time intervals $[0,\tau], [\tau, 2\tau], ..., [T-\tau, T]$, where $0<\tau\ll 1$ is a small time increment.
The increase of the saver's account at time $t+\tau$ can be therefore expressed as
\[
y_{t+\tau} = F_t^{\tau}(\theta, y_t, Z), \quad \hbox{where}\ \ Z\sim N(0,1),\ \hbox{and}
\]
\begin{equation}\label{CRE}
F_t^{\tau}(\theta, y_t, z)
= y_t \exp\left( 
\left(
\mu_t(\theta) - \beta_t -\frac{1}{2} \sigma_t^2(\theta) 
\right) \tau  + \sigma_t(\theta) z \sqrt{\tau}
\right)
+ \varepsilon \tau
\end{equation}
for $0 < t \leq T$. In the above expression for the function $F_t^\tau$ we have applied It\^{o}'s lemma (c.f. Kwok \cite{KWOK1998}) in order to generalize the discrete version of $F_t^1$ with $\tau=1$ for the case when $\tau$ is a sufficiently variable.

Let us denote by $V(t,y)$ the intermediate value (utility) function at time $t\in[0,T]$ corresponding to the saver's wealth of $y>0$ saved yearly salaries in her account. 
Making use of the integral definition of the expected value $\mathbb{E}_Z \left( W(t+1, F_t^1(\theta,y,Z) )\right)$ the continuous variant of the discrete backward value function for any choice of the small parameter $0 < \tau \ll 1$ takes the 
subsequent form
\begin{eqnarray*}
\label{CVF}
 &&V(T,y) =  U(y), \qquad t = T,\\
 &&V(t,y) = \max_{\theta \in\Delta_t} \int_{\R}V(t+\tau, F_t^{\tau}(\theta, y, z))f(z)\, dz,
\qquad 0 \leq t<t+\tau\le  T.
\end{eqnarray*}
Therefore for the limit $\tau \equiv dt \rightarrow 0^{+}$ we obtain
\begin{equation}
\label{maxEV}
\max_{\theta \in\Delta_t} \mathbb{E}\left( \frac{V(t+dt,y_{t+dt}) - V(t,y_t)}{dt}\;\Big| \;y_t = y\right) = 0 
\end{equation}

In general, we suppose that there exist functions $A_\E(\theta, t, y)$ and 
$B(\theta, t, y)$ such that the random process $y_t, t\in [0,T],$ is driven by the following stochastic differential equation
\begin{equation}\label{dySDE}
    dy_t = A_\E (\theta_t, t, y_t)dt + B(\theta_t, t, y_t)dW_t,
\end{equation} 
where $\{W_t, 0 \leq t \leq T\}$ is the Wiener process.
Then, by using It\^o's lemma (c.f. Kwok \cite{KWOK1998}) we obtain the expression for the 
differential $dV = V(t+dt, y_{t+dt}) - V(t,y_t)$ in the form of a function of two independent variables $t$ and $y$:
\begin{equation}\label{V-V}
\begin{split}
V(t+dt,y_{t+dt})&- V(t,y_t) \\
=& \Big[\frac{\partial V}{\partial t}(t,y_t)
+A_\E(\theta_t, t, y_t)\frac{\partial V}{\partial y}(t,y_t)
+\frac{1}{2}B^2(\theta_t, t, y_t)\frac{\partial^2 V}{\partial y^2}(t,y_t)\Big]dt \\
& + B(\theta, t, y_t)\frac{\partial V}{\partial y}(t,y_t)d W_t\,.
\end{split}
\end{equation}
Taking the conditional expectation of \eqref{V-V}, the second part in the expression above can be omitted since stochastic variables $B(\theta, y_t)\frac{\partial V}{\partial y}(t,y_t)$ and $d W_t$ are independent and $\mathbb{E}(d W_t)=0$. Hence
\begin{eqnarray*}
&&\mathbb{E}_Z \left(\frac{V(t+dt,y_{t+dt}) - V(t,y_t)}{dt} \;\biggl|\:y_t = y\right)\\
 &&\ \ \ \ = \frac{\partial V}{\partial t}(t,y) 
+ A_\E(\theta_t, t, y)\frac{\partial V}{\partial y}(t,y)+\frac{1}{2}B^2(\theta_t, t, y)\frac{\partial^2 V}{\partial y^2}(t,y).
\end{eqnarray*}
Letting $dt \rightarrow 0^{+}$, the maximum criterion in \eqref{maxEV} can be rewritten as follows
\begin{equation*}
\begin{split}
0 & =  \max_{\theta \in\Delta_t} \mathbb{E}\left(\frac{V(t+dt,y_{t+dt}) - V(t,y_t)}{dt}\;|\;y_t = y
\right) \\
& = \frac{\partial V}{\partial t}(t,y) + \max_{\theta
\in\Delta_t}\Big\{A_\E(\theta, t, y)\frac{\partial V}{\partial y}(t,y)
+\frac{1}{2}B^2(\theta, t, y)\frac{\partial^2 V}{\partial y^2}(t,y)\Big\}.
\end{split}
\end{equation*}

In our modeling what remains is to determine the concrete form of the functions $A_\E(\theta, t,y)$ and $B(\theta,t,y)$ driving the stochastic process (\ref{dySDE}) for $y_t$. Now, it follows from (\ref{CRE}) that, for an infinitesimal time increment $0<\tau=dt\ll 1$, we can apply It\^o's lemma to obtain the expression for the differential $dy_t = y_{t+dt} - y_t$:
\[
d y_t = \E dt + y_t \left( (\mu_t(\theta) - \beta_t) dt + \sigma_t(\theta) d W_t\right)
\]
where $dW_t = W_{t+dt} - W_t = Z \sqrt{dt}, Z\sim N(0,1)$, is the differential of the Wiener process. This way we have shown that the functions $A(\theta, t,y)$ and $B(\theta,t,y)$ driving the process (\ref{dySDE}) for the saver's wealth $y_t$ have the form:
\[
A_\E(\theta,t, y) = \varepsilon+\big[\mu_t(\theta) - \beta_t\big]y 
\qquad\hbox{and}\qquad B(\theta, t,y) = \sigma_t(\theta)y .
\]
\smallskip
In summary, we have derived the following result:
\medskip
\begin{thm}\label{HJBEthm}
The intermediate utility function $V = V(t,y)$ satisfies the following fully nonlinear partial differential Hamilton--Jacobi--Bellman equation:
\begin{equation}\label{HJBE1}
\frac{\partial V}{\partial t}(t,y) + \max_{\theta \in\Delta_t}\Big\{A_\E(\theta,
t, y)\frac{\partial V}{\partial y}(t,y)+\frac{1}{2}B^2(\theta, t,y)\frac{\partial^2 V}{\partial y^2}(t,y)\Big\} = 0
\end{equation}
and the terminal condition
$V(T,y) = U(y)$ for $y>0$ where
$A_\E(\theta,t, y) = \varepsilon+\big[\mu_t(\theta) - \beta_t\big]y$
and $B(\theta, t,y) = \sigma_t(\theta)y$.
\end{thm}

Henceforth, we shall make the following structural assumption on bond and stock average yields and their standard deviations:

\medskip
\noindent (H)\hskip 0.5truecm
$\left\{
\begin{matrix}
\hfill b_t & := \sigma_t^{(b)}[\sigma_t^{(b)}-\varrho_t \sigma_t^{(s)}]>0,\quad
a_t := [\sigma_t^{(s)}]^2+[\sigma_t^{(b)}]^2 - 2 \varrho_t \sigma_t^{(s)}\sigma_t^{(b)}> b_t\,,\cr
\Delta \mu_t & := \mu_t^{(s)} - \mu_t^{(b)} >0.
\hfill
\end{matrix} 
\right.$
\medskip

\noindent The last condition $a_t>b_t$ is equivalent to the inequality $\sigma_t^{(s)}-\varrho_t \sigma_t^{(b)}>0$ whereas the first condition for $b_t$ can be reformulated as the inequality $\sigma_t^{(b)}-\varrho_t \sigma_t^{(s)}>0$. Notice that these assumptions are generically fulfilled in stable financial markets since expected stock returns should outperform bond returns and the correlation $\varrho_t$ between stock and bond returns is negative in typical market situations (c.f. \cite{MU2004,KILIANOVA2006,MS2009}) and the discussion on model parameters in Section 4.

\medskip 
\subsection{Hamilton--Jacobi--Bellman equation for a super-optimal solution}

\medskip
Unfortunately, due to the restriction $\theta\in [0,1]$, the Hamilton--Jacobi--Bellman equation (\ref{HJBE1}) illustrates a difficult problem that cannot be, in general, solved analytically. Nevertheless, as we are approaching the problem of optimal pension fund portfolio construction we may relax the bound $\theta\le 1$ and allow the control parameter $\theta$ to vary over all nonnegative numbers, i.e. $\theta\in \tilde\Delta_t \equiv[0,\infty)$. Taking $\theta>1$ would correspond to the situation when we allow for the so-called short positions in bonds because $1-\theta < 0$ in that case. 

Instead of the Hamilton--Jacobi--Bellman equation (\ref{HJBE1}) we shall consider a  modified problem in which we allow the control parameter $\theta$ to vary over positive real numbers, 
i.e. 
\begin{equation}
\label{HJBE1-modified}
\frac{\partial V}{\partial t}(t,y) + \max_{\theta
\in[0,\infty)}\Big\{A_\E(\theta,t,y)\frac{\partial V}{\partial y}(t,y)
+\frac{1}{2}B^2(\theta,t,y)\frac{\partial^2 V}{\partial y^2}(t,y)\Big\} = 0.
\end{equation}

Under the assumption $\frac{\partial V}{\partial y}(t,y)>0, \frac{\partial^2 V}{\partial y^2}(t,y)<0$, the hypothesis (H) 
and taking into account the definitions  \eqref{mu} of $\mu_t(\theta)$ and $\sigma_t(\theta)$, the unique solution $\tilde{\theta}(t,y)$ to the implicit relationship  \eqref{STopt} is given by
\begin{subequations}
\begin{equation}
    \tilde{\theta}(t,y) = \frac{b_t}{a_t} 
- \frac{\Delta\mu_t}{a_t} \frac{\frac{\partial V}{\partial y} (t,y)}{y \frac{\partial^2 V}{\partial y^2} (t,y)}.
\end{equation}
\end{subequations}
Hence the Hamilton--Jacobi--Bellman equation \eqref{HJBE1-modified} can be rewritten as follows:
\begin{equation}\label{HJBE2}
\begin{split}
    0 = &\frac{\partial V}{\partial t}(t,y) + \big[\E+y(\mu_t^{(b)} - \beta_t+\frac{b_t}{a_t}\Delta\mu_t)\big]\frac{\partial V}{\partial y}(t,y)\\
&+\frac{1}{2a_t}[\sigma_t^{(b)}]^2[\sigma_t^{(s)}]^2(1-\varrho_t^2)y^2 \frac{\partial^2V}{\partial y^2}(t,y) - \frac{1}{2}\frac{(\Delta\mu_t)^2}{a_t}\frac{[\frac{\partial\;V}{\partial y}(t,y)]^2}{\frac{\partial^2 V}{\partial
y^2}(t,y)}.
\end{split}
\end{equation}

In what follows, we shall compare solutions $V(t,y)$ of the original Hamilton--Jacobi--Bellman equation  (\ref{HJBE1}) and the modified equation (\ref{HJBE2}). 

\medskip
\begin{thm}\label{HJBEcomparison}
Let $\Delta_t, \tilde\Delta_t \subset \R$ be two admissible sets such that $\Delta_t \subseteq \tilde\Delta_t$ for any time $t\in[0,T]$. Let $V(t,y)$ and $\tilde V(t,y)$ be solutions to the corresponding Hamilton--Jacobi--Bellman equations with admissible sets $\Delta_t$ and $\tilde \Delta_t$, i.e. 
\begin{eqnarray}
&&\frac{\partial V}{\partial t}(t,y) + \max_{\theta \in\Delta_t }\Big\{A_\E(\theta,t,y)\frac{\partial V}{\partial y}(t,y)
+\frac{1}{2}B^2(\theta,t,y)\frac{\partial^2 V}{\partial y^2}(t,y)\Big\} = 0, 
\label{prva}
\\
&&\frac{\partial \tilde V}{\partial t}(t,y) + \max_{\theta \in\tilde\Delta_t }\Big\{A_\E(\theta,t,y)\frac{\partial \tilde V}{\partial y}(t,y)
+\frac{1}{2}B^2(\theta,t,y)\frac{\partial^2 \tilde V}{\partial y^2}(t,y)\Big\} = 0,
\label{druha}
\end{eqnarray}
for $t\in [0,T), y>0,$ and satisfying the same terminal condition $V(T,y) = \tilde V(T,y) = U(y)$ for $y>0$. Then the solution $\tilde V$ of equation (\ref{druha}) is super-optimal for equation (\ref{prva}), i.e. 
\[
\frac{\partial \tilde V}{\partial t}(t,y) + \max_{\theta \in\Delta_t }\Big\{A_\E(\theta,t,y)\frac{\partial \tilde V}{\partial y}(t,y)
+\frac{1}{2}B^2(\theta,t,y)\frac{\partial^2 \tilde V}{\partial y^2}(t,y)\Big\} \le  0.
\]
Moreover, $V(t,y) \le \tilde V(t,y)$ for any $t\in [0,T], y >0$.
\end{thm}

\begin{proof}
The proof of the first part is rather simple and takes into account the inequality 
\[
\max_{\theta \in \Delta_t }\Big\{A_\E \partial_y \tilde V + \frac{1}{2} B^2 \partial^2_y \tilde V \Big\}
\le 
\max_{\theta \in \tilde \Delta_t }\Big\{A_\E \partial_y \tilde V + \frac{1}{2} B^2 \partial^2_y \tilde V \Big\} = - \partial_t\tilde V,
\]
because $\Delta_t \subseteq \tilde\Delta_t$. The second part easily follows from the parabolic comparison principle. Indeed, let $\theta^\ast(t,y)\in\Delta_t \subset\tilde\Delta_t$ be the optimal solution to (\ref{prva}), i.e. $\theta^\ast$ is the argument of the maximum operator in (\ref{prva}). Hence
\begin{eqnarray*}
&&\frac{\partial V}{\partial t}(t,y) + \Big\{A_\E(\theta^\ast(t,y),t,y)\frac{\partial V}{\partial y}(t,y) +\frac{1}{2}B^2(\theta^\ast(t,y),t,y)\frac{\partial^2 V}{\partial y^2}(t,y)\Big\} = 0, 
\\
&&\frac{\partial \tilde V}{\partial t}(t,y) + \Big\{A_\E(\theta^\ast(t,y),t,y)\frac{\partial \tilde V}{\partial y}(t,y)
+\frac{1}{2}B^2(\theta^\ast(t,y),t,y)\frac{\partial^2 \tilde V}{\partial y^2}(t,y)\Big\} \le 0.
\end{eqnarray*}
Therefore, applying the parabolic comparison principle (see e.g. \cite{PROTTER}) we conclude the inequality $V(t,y) \le \tilde V(t,y)$ for any $t\in[0,T]$ and $y>0$, as claimed.
\end{proof}

The above theorem enables us to refer to a solution $V$ to the modified Hamilton--Jacobi--Bellman equation (\ref{HJBE1-modified}) to as a super-optimal solution to the original equation  (\ref{HJBE1}).

\section{Transformation to a quasi-linear parabolic equation}

In what follows, we shall simplify our model by assuming all the model parameters to be constant with respect to time $t\in [0,T]$, i.e. $\mu_t^{(b)}=\mu^{(b)}, \mu_t^{(s)}=\mu^{(s)}, \sigma_t^{(b)}=\sigma^{(b)}, \sigma_t^{(s)}=\sigma^{(s)}, \varrho_t=\varrho, \beta_t=\beta$. Consequently, $a_t=a, b_t=b, \Delta\mu_t =\Delta\mu$. 

Following ideas borrowed from the recent paper by Abe and  Ishimura \cite{ISHIMURA2008} we introduce the Riccati-like transformation
\begin{equation*}
\varphi (t,y) = -\frac{\frac{\partial^2 V}{\partial y^2}(t,y)}{\frac{\partial V}{\partial y}(t,y)}, 
\end{equation*}
and the auxiliary function
$g(t,y) = \E+ \frac{(\Delta\mu)^2}{2a}\frac{1}{\varphi(t,y)}+\alpha y - \frac{1}{2}c^2 y^2
\varphi(t,y)$ where
\begin{equation}
\label{alpha}
    \alpha = \mu^{(b)} - \beta+\frac{b}{a}\Delta\mu
\qquad\hbox{and}\qquad c =
\sigma^{(b)}\sigma^{(s)}\sqrt{\frac{1-\varrho^2}{a}}.
\end{equation}
Then equation \eqref{HJBE2} can be transformed to the form
\[
\frac{\partial V}{\partial t}(t,y) +g(t,y)\frac{\partial V}{\partial y}(t,y) = 0.
\]
In terms of the transformed function $\varphi$, 
\begin{equation}
    \frac{\partial \varphi}{\partial t}(t,y) = \frac{\partial}{\partial y}
\left( 
\frac{\partial g}{\partial y}(t,y) - \varphi(t,y)g(t,y)
\right)
\end{equation}
Moreover, the unique solution \eqref{HJBE1-modified} to the implicit relationship \eqref{IE} is given by
\begin{equation}\label{ttm}
    \tilde{\theta}(t,y) = \frac{b}{a} +\frac{\Delta\mu}{y \varphi (t,y)}.
\end{equation}
Let us introduce the following change of independent variables
\[
x = \ln y,  \qquad  s = T-t\qquad \hbox{where}\qquad  0 \leq s\leq T,\ \ x\in\R, 
\]
and the transformation:
\begin{equation}
\label{psi-phi}
\psi(s,x) = \gamma y \varphi(t,y)\qquad\hbox{where}\quad \gamma =\frac{c\sqrt{a}}{\Delta\mu}.
\end{equation}
Then the original HJB equation \eqref{HJBE2} stated  for the intermediate utility
function $V(t,y)$ can be reformulated for the function $\psi(s,x)$ as
follows:
\begin{subequations}
\begin{equation}\label{HJBE3e}
\frac{\partial \psi}{\partial s}=  \frac{c^2}{2} \frac{\partial}{\partial x}
\left(
\left[1+\frac{\partial}{\partial x}\right]\left(\psi - \frac{1}{\psi}\right) +\psi\left(
\frac{2}{c^2}(\E e^{-x}+\alpha) - \frac{\psi}{\gamma}
\right)
\right)
\end{equation}
for $s\in (0,T],\ x\in \R$. The solution $\psi$ is subject to the initial condition
\begin{equation}\label{HJBE3c}
\psi(0,x) = -\gamma \frac{U\DD(e^x)}{U\D(e^x)}e^x, \qquad\mbox{for}\ \  x\in \R.
\end{equation}
\end{subequations}

Therefore the optimal $\tilde{\theta}$ arising from the implicit equation \eqref{IE} and originally given by \eqref{ttm}, now expressed in terms of new variables $(s,x)$ and $\psi(s,x)$ takes the subsequent form
\begin{equation}\label{ttmm}
\tilde{\theta}(t,y) = \frac{b}{a}  + \frac{\gamma \Delta\mu }{a \psi(T-t, \ln
y)}.
\end{equation}

\medskip
\begin{remark}
It must be remarked that the HJB equation \eqref{HJBE3e}--\eqref{HJBE3c} is not applicable to the Life-cycle model  (see e.g. \cite{BODIE1992}). Notice that in the Life-cycle model, the stock to bond ratio $\tilde\theta$ is designed in a way it depends only on the age $a$ of a future pensioner. A typical choice for $\tilde\theta$ in the Life cycle model is $\tilde\theta=1- a/100$. Therefore such a stock to bond ratio $\tilde\theta$ is independent of $y$ variable. 
Taking into account \eqref{ttmm}, we obtain 
\begin{equation}\label{LCc}
\frac{\partial\psi}{\partial x} (s,x) = 0 , \qquad\mbox{for}\ \  x\in \R \ \mbox{and} \ s \in [0,T],
\end{equation}
i.e. $\psi$ is constant in the $x$- variable. Then equation \eqref{HJBE3e} can be reworded to:
\begin{equation}\label{HJBEeLC}
\frac{\partial \psi}{\partial s} = \E e^{-x}\psi, \qquad\mbox{for}\ \  x\in \R \ \mbox{and} \ s \in [0,T].
\end{equation}
Clearly, the above equality is impossible as the function $\psi$ depends on $s$ only. As a consequence, the Life-cycle model can not be described by the dynamic stochastic optimization model.
\end{remark}

\subsection{The constant relative risk aversion (CRRA) utility function}

In this part we discuss a suitable choice of the utility function $U$. 
We must emphasize that the utility function may vary across investors as it 
represents their attitude to risk. According to Arrow and Pratt the attitude
to risk can be expressed in terms of the so-called 
coefficient of relative risk aversion defined as 
$C(y)=-y U''(y)/U'(y)$. A constant relative risk aversion $C(y)\equiv d>0$ 
for every $y>0$ would imply that an investor tends to hold a constant proportion 
of his wealth in any class of risky assets as the wealth varies. The reader
is refereed to a vast economic literature addressing the problem of a proper
choice of investor's  utility function 
 (see e.g. Friend \& Blume \cite{FRIEND1975}, Pratt \cite{PRATT1964} and Young \cite{YOUNG1990}). 

In the case of a constant relative risk aversion $C(y)\equiv d>0$ 
an increasing utility function $U$ is uniquely (up to an multiplicative and
additive constant) given by
\begin{equation} \label{isoel}
U(y) = -y^{1-d}\quad\textstyle{\rm if}\; d>1\,,\ \ \ 
U(y) =   \ln(y) \quad\textstyle{\rm if}\; d=1\,, \ \ \ 
U(y) =  y^{1-d} \quad\textstyle{\rm if}\; d<1\,.
\end{equation}
The coefficient $d$ of relative risk aversion  plays an important role in many fields of theoretical economics. There is a wide consensus that the value should be less than 10 (see e.g Mehra and Prescott \cite{MEHRA1985}). In our numerical experiments we considered values of $d$ close to 9. But it could be also lower for lower equity premium. It is worth to note that the CRRA function is a smooth, increasing and strictly concave function for $y>0$.

For the purpose of our forthcoming analysis we consider the utility function $U(y)$ (i.e.
the terminal condition for \eqref{HJBE}) of the form
\begin{equation}\label{util}
    U(y) = -y^{1-d} \qquad \hbox{where} \qquad d>1. 
\end{equation}
The function $U$ is a smooth strictly increasing concave function. Now it
should be obvious that the power like behavior of the utility function $U(y)=-y^{1-d}$ 
leads to the constant initial condition \eqref{HJBE3c}, i.e. 
\begin{equation}
    \psi(0,x) = \gamma d, \quad \hbox{for any} \ \ x\in \R\,.
\end{equation}

\subsection{Construction of appropriate sub- and super-solutions}

In this part we shall derive effective lower and upper bounds of a solution $\psi$ to equation (\ref{HJBE3e}). We restrict our attention to the case when the initial condition $\psi(0,x)$ is constant, i.e. the function $U$ is the CRRA utility function of the form $ U(y) = -y^{1-d}$ for some $d>0$. 

Equation (\ref{HJBE3e}) represents a fully nonlinear parabolic equation of the form
\begin{equation}
\label{fullynonlinear}
\frac{\partial \psi}{\partial s}=\mathcal{H}(t,y,\psi,\frac{\partial \psi}{\partial y}, \frac{\partial^2 \psi}{\partial y^2}).
\end{equation}
Notice that the right hand side of (\ref{fullynonlinear}) is a strictly parabolic 
operator such that
\[
\frac{\partial \mathcal{H}}{\partial q} (t,y,\psi,p,q) = \frac{c^2}{2}\big( 1+\frac{1}{\psi^2}\big) \ge \frac{c^2}{2}>0. 
\]
In what follows, we shall construct positive sub- and super-solutions to the fully nonlinear parabolic equation  (\ref{fullynonlinear}). The idea behind the construction of suitable sub- and super-solution is rather simple and it takes into account the form of the terminal condition $U(y)$ for the intermediate 
utility function $V(t,y)$ at $t=T$. With regard to the presence of an advective term in the equation for the function $V$ it is therefore reasonable to compare $V(t,y)$ and the translated terminal function $U(y+\varsigma(T-t))$ where $\varsigma$ is a positive function to be determined later. In terms of the transformed function $\psi$ a suitable candidate for a sub- or super-solution to  (\ref{fullynonlinear}) can be therefore sought in the
form: 
\begin{equation}
\label{subsuper}
\psi(s,x) = -\gamma e^x \frac{U^{\prime\prime}(e^x + \varsigma(s))} {U^{\prime}(e^x + \varsigma(s))}
\end{equation}
where $\varsigma(s)\ge 0$, $ \varsigma(0)=0,$ is a smooth function to be determined later. Assuming $ U(y) = -y^{1-d}$ we have
\[
\psi(s,x) = \gamma d \frac{1}{1 + \varsigma(s) e^{-x}}.
\]
Next we calculate
\[
\frac{\partial \psi}{\partial s}= -  \gamma d  \frac{\varsigma^\prime(s) e^{-x}}{(1 + \varsigma(s) e^{-x})^2},
\ \ 
\frac{\partial \psi}{\partial x}=  \gamma d  \frac{\varsigma(s) e^{-x}}{(1 + \varsigma(s)
e^{-x})^2},
\]
where $(.)^\prime$ stands for the derivative with respect to $s$. Now it is an easy calculus to verify the following identity:
\[
\frac{\partial}{\partial x}\left( 
\left[1+\frac{\partial}{\partial x}\right]\left(\psi-\frac{1}{\psi}\right) \right)= \gamma d
\frac{\varsigma(s)^2 e^{-2x}}{(1 + \varsigma(s) e^{-x})^3}.  
\]
Denote by $\mathcal{A} = 2\alpha/c^2, \mathcal{B}=2/c^2, \mathcal{C} = \Delta\mu
/(c\sqrt{a})=1/\gamma$. Then the lower order term in the right hand side $\mathcal{H}$ of equation (\ref{fullynonlinear}) has the form 
\[
\frac{\partial}{\partial x}\left( 
\left[
\mathcal{A} + \E \mathcal{B} e^{-x}  - \mathcal{C} \psi
\right] \psi 
\right) 
=\gamma d  \frac{ e^{-x}}{(1 + \varsigma(s) e^{-x})^2} \left(
\mathcal{A} \varsigma + \E \mathcal{B}   -2\mathcal{C}\gamma  d   \frac{ \varsigma(s) }{1 + \varsigma(s) e^{-x}}
\right).
\]
Hence, 
\[
\frac{(1 + \varsigma(s) e^{-x})^2}{\gamma  d e^{-x}} \mathcal{H} = 
\frac{ 2 \varsigma^2 e^{-x}}{1 + \varsigma(s) e^{-x}} + \mathcal{A} \varsigma + \E \mathcal{B}   
-2\mathcal{C} \gamma  d   \frac{ \varsigma }{1 + \varsigma(s) e^{-x}}.
\]
Clearly, the following two simple inequalities hold:
\[
0\le \frac{ \varsigma^2 e^{-x} }{1 + \varsigma(s) e^{-x}} \le \varsigma,\quad 0\le \frac{ \varsigma }{1 + \varsigma(s) e^{-x}} \le \varsigma
\]
for any $\varsigma \ge 0$ and $x\in \R$. 
Let $\underline{\varsigma}(s)$ and $\overline{\varsigma}(s), s\ge 0,$ be solutions to the linear
ODEs:
\begin{equation}
\label{rovnicenag}
\begin{matrix}
- \underline{\varsigma}^\prime(s) =\frac{c^2}{2}\left( 
(\mathcal{A} - 2\mathcal{C}\gamma  d) \underline{\varsigma}(s) - \E \mathcal{B}
\right),\hfill &
\qquad \underline{\varsigma}(0 )=0,
\cr
\cr
-  \overline{\varsigma}^\prime(s) = \frac{c^2}{2}\left( 
(2+\mathcal{A}) \overline{\varsigma}(s) - \E \mathcal{B}
\right),\hfill &
\qquad \overline{\varsigma}(0 )=0.
\end{matrix}
\end{equation}
Then it is a straightforward calculus to verify that both $\underline{\varsigma}(s)$ as well as $\overline{\varsigma}(s)$ are nonnegative and the functions
\[
\underline{\psi}(s,x) = \frac{ \gamma  d}{1 + \underline{\varsigma}(s) e^{-x}},\quad 
\overline{\psi}(s,x) = \frac{ \gamma  d}{1 + \overline{\varsigma}(s) e^{-x}}, 
\]
are sub- and super-solutions to the strictly parabolic nonlinear equation (\ref{fullynonlinear}), i.e. 
\[
\partial_s\underline{\psi} 
\le \mathcal{H}(t,y,\underline{\psi},\frac{\partial \underline{\psi}}{\partial y}, \frac{\partial^2 \underline{\psi}}{\partial y^2}),
\qquad 
\partial_s\overline{\psi} 
\ge \mathcal{H}(t,y,\overline{\psi},\frac{\partial \overline{\psi}}{\partial y}, \frac{\partial^2 \overline{\psi}}{\partial y^2}),
\]
satisfying the same constant initial condition
$\overline{\psi}(0,x)= \underline{\psi} (0,x) = \gamma  d$ for any $x\in \R$.

Applying the parabolic comparison principle for strongly parabolic equations
(see e.g. \cite{PROTTER}) we deduce the following comparison result:

\begin{thm}
\label{th:CompPsi}
The solution $\psi(s,x)$ to the fully nonlinear parabolic equation (\ref{fullynonlinear}) satisfies the following inequalities:
\[
0<  \frac{ \gamma  d}{1 + \underline{\varsigma}(s) e^{-x}} \le \psi(s,x) \le  \frac{ \gamma  d}{1 + \overline{\varsigma}(s) e^{-x}}<\infty
\]
for any $s\in (0,T)$ and $x\in\R$ where the functions $\underline{\varsigma}(s)$, $\overline{\varsigma}(s)$, $s\ge 0,$ are the
unique solutions to the ODEs (\ref{rovnicenag}), i.e.
\begin{equation}
\label{gfunkcie}
\overline{\varsigma}(s) = \varepsilon (1-\exp(-\overline{\lambda}s))/\overline{\lambda},
\qquad
\underline{\varsigma}(s) = \varepsilon (1-\exp(-\underline{\lambda}s))/\underline{\lambda},
\end{equation}
where $\overline{\lambda}=\alpha-c^2 d$ and
$\underline{\lambda}=\alpha+c^2$, resp.
\end{thm}

Now, taking into account the relationships (\ref{ttm}), (\ref{psi-phi}) and
integrating the above inequalities for the function $\psi(s,x)$ with respect
to $x$  we obtain the bound for the intermediate utility function $V(t,y)$.

\begin{thm}
\label{th:CompV}
The intermediate utility function $V(t,y)$ satisfies the following inequalities:
\[
-(y+  \overline{\varsigma}(T-t))^{1-d} \le V(t,y) \le  -(y+  \underline{\varsigma}(T-t))^{1-d}
\]
for any $t\in[0,T]$ and $y>0$ where the functions $\underline{\varsigma}(s)$, $ \overline{\varsigma}(s)$, $s\ge 0,$ are given by (\ref{gfunkcie}).
Moreover, the function $y\mapsto V(t,y)$ is strictly increasing and strictly
concave function, 
\[
\frac{\partial V}{\partial y} (t,y) >0 \quad \hbox{and}\ \ 
\frac{\partial^2 V}{\partial y^2} (t,y) <0  \quad \hbox{for any}\ \  t\in[0,T],\  y>0.
\]
\end{thm}

Applying the previous theorem and using the expression (\ref{ttm}) for the optimal value 
$\hat\theta$ we are now in a position to state useful bounds for the optimal stock to bond
proportion in the optimal portfolio.

\begin{thm}\label{odhad-theta}
The optimal value $\hat\theta(t,y)$ describing the optimal stock to bond proportion 
in the optimal portfolio satisfies the inequalities:
\begin{equation}
\label{optimal-theta-bounds}
\frac{b}{a} + \frac{\Delta\mu}{a d}\left( 1+\frac{\overline{\varsigma}(T-t)}{y}\right)
\le 
\hat\theta(t,y) 
\le
\frac{b}{a} + \frac{\Delta\mu}{a d}\left(1+\frac{\underline{\varsigma}(T-t)}{y}\right).
\end{equation}
Moreover, as a consequence of the hypothesis (H), we have $\hat\theta(t,y)
\ge 0$ for any $y>0$ and $t\in[0,T]$. If, in addition, the coefficient of
the relative risk aversion satisfies $d \ge \Delta\mu/(a-b)$ then for the terminal
value $\hat\theta(T,y)$ we have $\hat\theta(T,y)=b/a + \Delta\mu/(a d) \le 1$. 
\end{thm}

\medskip
It is worth to note that the condition $d \ge \Delta\mu/(a-b)$ can be expressed in terms of average stock/bond returns and their volatilities as follows:
\begin{equation}
\label{dNerovnost}
d\ge \frac{\mu^{(s)} - \mu^{(b)}}{\sigma^{(s)} (\sigma^{(s)} - \varrho \sigma^{(b)} )}.
\end{equation}
In the concrete example of the Slovak fully funded pension fund system (see Section 4) the above constraint reads as $d>1.78$. It means that it is fulfilled in typical market data situations for an individual saver having the coefficient of relative risk aversion $d$ greater than $1.78$. 

\medskip
In order to construct a solution $\psi$ to the problem \eqref{HJBE3e} let us rewrite $\psi(s,x)$ in terms of the asymptotic series with respect to the small parameter $\E$ as follows:
\begin{equation}\label{AS}
    \psi(s,x) = \sum_{n=0}^{\infty}\E^n \psi_n(s,x).
\end{equation} 
The parameter $0<\varepsilon\ll 1$ can be considered as a small parameter. In practical applications of the dynamic stochastic accumulation model, the value of $\varepsilon$ is close to $0.09$ (Slovak pension saving system discussed in Section 4) or $\varepsilon\approx 0.14$ (Bulgarian pension saving system \cite{VULKOV09}). As it could be obvious from the discussion in the next paragraph, the parameter $\varepsilon$ is a natural candidate for a expansion parameter because  we know the explicit solution $\psi_0(s,x)$ for equation 
(\ref{HJBE3e}) for the vanishing parameter $\varepsilon=0$.

\subsection{No contributions -- the zeroth order approximation}
First of all we pay special attention to the first term in the Taylor expansion above, 
$\psi_0(s,x)$. Recall that due  to the power like character of the utility function
$U(y)=-y^{1-d}$, we have
\begin{equation}\label{psi0}
    \psi_0(s,x) =\gamma  d \quad\hbox{for any}\ s\in [0,T],\ x\in\R.
\end{equation} 
Indeed, any constant function in the $x$ variable is a solution to
(\ref{HJBE3e}) with $\varepsilon=0$. Moreover, for the power like function
$U$, the initial condition (\ref{HJBE3c}) is also constant and it is equal to
$\gamma d$. Therefore $\psi_0(s,x) =\gamma  d$ for any $s\in [0,T]$ and
$x\in\R$.

Let us consider the limiting case where there are no defined contributions, i.e. $\E=0$. Then the solution $\psi(s,x)$ coincides with $\psi_0$ and this is why the solution $\psi$ is constant in time $s$ and spatial variable $x$. As a consequence we obtain, for $\E=0$ that the optimal stock to bonds proportion is also constant, i.e.
\[
\hat\theta(t,y) = \frac{\sigma^{(b)}(\sigma^{(b)} -\varrho \sigma^{(s)}) +\frac{\mu^{(s)}-\mu^{(b)}}{d}}{[\sigma^{(s)}]^2 + [\sigma^{(b)}]^2 - 2\varrho \sigma^{(s)} \sigma^{(b)}}
\]
for any $t\in[0,T]$ and $y>0$. This observation is in agreement with Merton and Samuelson's result (c.f. \cite{MERTON1969,SAM1969,MERTON1974,MERTON1971}) stating that the stock to bond proportion is constant and it depends on saver's risk aversion only. Notice that $\theta(t,y)\in[0,1]$ provided that the condition (\ref{dNerovnost}) is fulfilled.

\subsection{The first order approximation}
In order to roughly approximate the function $\psi(s,x)$ for small enough values of the parameter $\E$, we use both the constant and the linear terms corresponding to the asymptotic expansion \eqref{AS} to get
\begin{equation}\label{TElin}
    \psi(s,x) = d\gamma+\E\psi_1(s,x) + O(\E^2) \quad\hbox{as} \ \ \E\to 0^+
\end{equation}
where $\psi_1(s,x)$ is an unknown function to be specified. 
Replacing the original function $\psi(s,x)$ by its linear approximation \eqref{TElin} above in the problem \eqref{HJBE3e}, the Hamilton--Jacobi--Bellman equation for the function $\psi_1(s,x)$ takes the ensuing form 
\begin{subequations}\label{HJBEpsi1}
\begin{equation}\label{HJBEpsi1e}
\begin{split}
    \frac{\partial\psi_1}{\partial s}(s,x)= 
\frac{c^2}{2}\Big[1+&\frac{1}{\psi_0^2}\Big]\frac{\partial^2 \psi_1}{\partial\;x^2}(s,x) 
+\frac{c^2}{2}\Big[1+\frac{1}{\psi_0^2}+\frac{2\delta}{c^2}\Big]\frac{\partial \psi_1}{\partial x}(s,x)
- \psi_0 e^{-x}
\end{split}
\end{equation}
for any $s\in(0,T)$ and $x\in\R$. The solution is subject to the initial condition
\begin{equation}\label{HJBEpsi1c}
    \psi_1(0,x) = 0 \qquad\hbox{for any}\ \  x\in\R,
\end{equation}
where $d\gamma$ is replaced by the constant $\psi_0$ and, for abbreviation, 
$\delta$ stands for the following expression:
\begin{equation}
    \delta  = \alpha-d c^2.
\end{equation}
\end{subequations}
The unique solution of the Cauchy problem (\ref{HJBEpsi1e}) can be found in a separable form:
\begin{equation}
\label{psi1-form}
\psi_1(s,x) = \Phi_1(s) e^{-x}, \quad \hbox{where} \ \ 
\Phi_1(s)= d \gamma \frac{ e^{-\delta s} -1}{\delta}\quad \hbox{for any} \
s\in(0,T), x\in\R.
\end{equation}

Thus if the higher order terms in \eqref{TElin} are omitted, the explicit approximate solution of the problem \eqref{HJBEpsi1} can be written as 
\begin{equation}\label{psi1}
    \psi(s,x) = \psi_0 + \E \psi_1(s,x)+ O(\E^2) 
= d\gamma \Big\{1 + \E \frac{e^{-\delta s}-1}{\delta}e^{-x}\Big\} +
O(\E^2)
\end{equation}
as $\E\to 0^+$. Expanding the optimal value $\hat\theta(t,y)$ and using the
formula (\ref{ttm}) we obtain the first order approximation of
$\hat\theta(t,y)$ in the form
\begin{equation}\label{oThetaLin}
    \hat{\theta}(t,y) = \frac{b}{a} +\frac{\Delta\mu}{a d}\Big[1 
+ \frac{\E}{y}\frac{1-e^{-\delta(T-t)}}{\delta}\Big] + O(\E^2).
\end{equation}
Since $\delta=\alpha-d c^2$ is the same constant as $\overline{\lambda}$
entering the expression for the lower bound of $\hat\theta$ (see
(\ref{optimal-theta-bounds})) we may conclude that the first order
approximation of the optimal value of $\hat\theta$ coincides with its lower
bound given by (\ref{optimal-theta-bounds}) and (\ref{gfunkcie}).

\subsection{The second term approximation}
For the reason of better approximation of the function $\psi(s,x)$ for small enough values of the parameter $\E$, now we make use the Taylor expansion \eqref{AS} up to the second order term
\begin{equation}\label{TEqua}
    \psi(s,x) = \psi_0+\E\psi_1(s,x) + \E^2\psi_2(s,x)+O(\E^3)
\end{equation}
as $\E\to 0^+$. Recall that we already have computed the first two terms
$\psi_0,\psi_1$ in the expansion. Namely, 
$\psi_0 = d\gamma, \psi_1(s,x) = \Phi_1(s) e^{-x}$, 
where $\Phi_1(s)=\frac{d\gamma}{\delta}\left(e^{-\delta
s}-1\right)$.
The function $\psi_2(s,x)$ is an unknown second order expansion of the function $\psi$ to
be determined.

Inserting the quadratic approximation \eqref{TEqua} of the function $\psi(s,x)$ 
into equation \eqref{HJBE3e} and calculating all the terms of the order $O(\E^2)$ we conclude that the function $\psi_2(s,x)$ is a solution to the following linear parabolic equation:
\begin{equation}\label{HJBEpsi2e}
\begin{split}
    \frac{\partial{\psi_2}}{\partial s}(s,x)
= \frac{c^2}{2}\Big[1+\frac{1}{\psi_0^2}\Big] \frac{\partial^2\psi_2}{\partial x^2}(s,x) 
+ \frac{c^2}{2}\Big[1+\frac{1}{\psi_0^2}+\frac{2\delta}{c^2}\Big]\frac{\partial \psi_2}{\partial x}(s,x) 
+ e^{-2x}\xi_2(s)
\end{split}
\end{equation}
satisfying the initial condition $\psi_2(0,x) = 0$ for $x\in\R$, where
\begin{equation}
\label{zeta}
\xi_2(s) = c^2 \left( \frac{1}{\gamma} - \frac{1}{(d\gamma)^3}\right)
\Phi_1^2(s) - 2 \Phi_1(s).
\end{equation}
The explicit  solution of the problem \eqref{HJBEpsi2e} can be written in a closed form: 
\begin{equation}
\label{psi2}
    \psi_2(s,x) = e^{-2x}\int_0^s\xi_2(z)\exp\left( c^2 (s-z)[1+\frac{1}{\psi_0^2}-\frac{2\delta}{c^2}]\right)dz\,.
\end{equation}
The integral appearing in (\ref{psi2}) can be explicitly computed and it can be expressed as a linear combination of of three exponential functions in the $s$ variable. 

\subsection{The general asymptotic series solution for $n\ge 3$}

From the straightforward analysis of the asymptotic expansion \eqref{AS} first two terms one can deduce not only the separability property of the terms with respect to both variables $s$ and $x$ but the exponential contribution of the variable $x$ to the solution's $n$th term. Thus in order to determine the $n$th term of the asymptotic expansion of the solution to the problem \eqref{HJBE3e} let us reformulate the original solution expansion \eqref{AS} in terms of the special asymptotic series with respect to the small parameter $\E$
\begin{equation}\label{ASS}
    \psi(s,x) = \sum_{n=0}^{\infty}\E^n \psi_n(s,x) = \sum_{n=0}^{\infty}\E^n\Phi_n(s)e^{-nx}
\end{equation}with the constant zero term $\psi_0(s,x) = \Phi_0(s) = d\gamma$ Then the general $n$th term of the solution asymptotic expansion can be determined recursively by the following linear non-homogeneous first-order ordinary differential equation:
\begin{eqnarray}
\label{nTerm}
\Phi_n\D(s) &-& 
\frac{c^2}{2}\left(
n(n-1)(1+\frac{1}{\Phi_0^2}) - \frac{2}{c^2} n \alpha + 2 n d 
\right)\Phi_n(s) \nonumber
\\
&=& \frac{c^2}{2}\frac{n(n-1)}{\Phi_0}\sum_{k=1}^{n-1}\Phi_{n-k}(s)\Omega_k(s) 
- n \Phi_{n-1}(s)+\frac{n}{\gamma} \frac{c^2}{2}  \sum_{k=1}^{n-1}\Phi_{n-k}(s)\Phi_k(s)
\end{eqnarray}
for $n\geq 1$ where $\Omega_k(s)$ is represented by the recurrent formula below
\begin{equation}\label{nTermOmega}
    \Omega_n(s) = -\frac{1}{\Phi_0}\sum_{k=0}^{n-1}\Omega_k(s)\Phi_{n-k}(s), \qquad \Omega_0 = \frac{1}{\Phi_0}.
\end{equation}
Then by solving the recurrent differential equation \eqref{nTerm} for the $n$th term we obtain
\begin{equation}
\Phi_n(s) = \int_0^s\xi_n(z)
\exp\left(
\frac{c^2}{2}
\big[
n(n-1) (1+\frac{1}{\Phi_0^2})-\frac{2 n \alpha}{c^2}+ 2 d n
\big](s-z)
\right) dz
\end{equation}
where
\begin{equation}
    \xi_n(z) = \frac{c^2}{2} \frac{n(n-1)}{\Phi_0}\sum_{k=1}^{n-1}\Phi_{n-k}(z)\Omega_k(z) - n \Phi_{n-1}(z)+\frac{n}{\gamma}\frac{c^2}{2}\sum_{k=1}^{n-1}\Phi_{n-k}(z)\Phi_k(z).
\end{equation}

\subsection{Qualitative behavior of the optimal value $\hat\theta$}

In this section we will be concerned with some useful analytic properties of the value $\hat\theta(t,y)$. We restrict our attention to the first order approximation of $\hat\theta$ given by the leading terms in (\ref{oThetaLin}), i.e. 
\begin{equation}
\label{oThetaLin1}
\hat{\theta}(t,y) = \frac{b}{a} +\frac{\Delta\mu}{a d}\Big[1 
+ \frac{\E}{y}\frac{1-e^{-\delta(T-t)}}{\delta}\Big].
\end{equation}
We will show that even this first order approximation is capable of capturing all interesting phenomena that are present in our dynamic stochastic optimization problem for optimal choice of the stock to bond proportion in pension fund portfolios.

It is worth to note that $(1-\exp(-\delta (T-t)))/\delta \ge 0$ for any $\delta\in\R$. Hence it is easy to verify that
\begin{equation}
\label{y-t-nerovnosti}
\frac{\partial\hat\theta}{\partial y} (t,y) < 0 \qquad\hbox{and}\ \ 
\frac{\partial\hat\theta}{\partial t} (t,y) < 0
\end{equation}
for any $t\in [0,T), y>0$.
It means that the optimal stock to bonds proportion is a decreasing function with respect to time $t$ as well as to the amount $y>0$ of yearly saved salaries. 

\subsubsection{Sensitivity of the optimal value with respect to the small parameter $\E$}
First we consider the dependence of the optimal value $\hat\theta$ on the small parameter $\E>0$ representing the percentage of the transfer of yearly salary to pensioner's account. It follows from (\ref{oThetaLin1}) that 
\begin{equation}
\label{theta-derE}
\frac{\partial\hat\theta}{\partial\E} (t,y) = 
\frac{\Delta\mu}{a d y} \frac{1-e^{-\delta(T-t)}}{\delta} >0
\end{equation}
for any $0\le t<T$ and $y>0$. As a consequence of the above inequality we may deduce that, the optimal value $\hat\theta$ is an increasing function in $\E$. Taking into account the possible application in the dynamic accumulation pension saving model, we can conclude that the higher percentage $\E$ of salary transferred each year to a pension fund would lead to higher optimal stock to bond proportion $\hat\theta$. 

Supposing that the percentage $\E$ represents investor's net contributing ratio, i.e. $$\E = (1-\kappa)\tilde{\E}$$
where $\kappa$ are managing costs, a regular fee charged by the the pension fund management institutions administering investor's private pension account and $\tilde{\E}$ stands for the gross salary ratio of the financial transfer. In the Slovak pension system one has $\tilde\E=0.09$ and $\kappa=0.01$, i.e. $\E = 0.0891$.
Evidently, \eqref{theta-derE} results in
\begin{equation}
\label{theta-derK}
\frac{\partial\hat\theta}{\partial \kappa} (t,y) = 
-\tilde{\E} \frac{\Delta\mu}{a d y} \frac{1-e^{-\delta(T-t)}}{\delta} <0
\end{equation}
for any $0\le t<T$ and $y>0$. It means that the increase in  managing costs implies decrease in the stock to bond proportion, as expected.

\subsubsection{Dependence of the optimal value on the saver's risk aversion}
Our next sensitivity analysis is focused on the dependence of the optimal value $\hat\theta$ on the coefficient $d$ measuring saver's risk aversion. Again, it follows from (\ref{oThetaLin1}) that 
\begin{equation}
\label{theta-derd}
\frac{\partial\hat\theta}{\partial d} (t,y)= 
-\frac{\Delta\mu}{a d^2} 
\left[1+ \frac{\E}{y} \omega(T-t)
\right],\ \  \hbox{where}\ \ \omega(s) = \frac{1-e^{-\delta s}}{\delta}\left(1-\frac{d c^2}{\delta} 
\right) 
+ s d c^2 \frac{e^{-\delta s}}{\delta}.
\end{equation}
Since $\omega(0)=0$ and $\omega^\prime(s) = (1- s d c^2) \exp(-\delta s)$ we may conclude that $\omega(T-t)$ is positive, and, consequently  $\frac{\partial\hat\theta}{\partial d} (t,y)<0$ provided that the coefficient of the saver's relative risk aversion satisfies 
\begin{equation}
\label{condition-d}
d \le \frac{1}{c^2 T} \equiv \frac{[\sigma^{(s)}]^2 + [\sigma^{(b)}]^2 - 2 \varrho \sigma^{(s)} \sigma^{(b)}}{[\sigma^{(s)}]^2 [\sigma^{(b)}]^2 (1-\varrho^2) T }.
\end{equation}
Notice that the fraction $\frac{1}{c^2 T}\approx 306$ in our market data parameter settings discussed in the next section. In other words condition (\ref{condition-d}) is fullfiled for typical values of saver's risk aversion coefficient $d\approx 10$. In summary, we have shown that the optimal stock to bond proportion $\hat\theta$ is a decreasing function with respect to the saver's risk aversion. In other words, higher risk aversion leads to less amount of stocks in saver's portfolio, as expected.

\subsubsection{Sensitivity with respect to average stock returns}
Another important sensitivity analysis is concerned with the dependence of the optimal value $\hat\theta$ on the performance of the stock part of the portfolio. More precisely, we shall study the dependence of $\hat\theta$ on the average stock return $\mu^{(s)}$. Such an analysis can be useful when the stock market is unstable and is exposed to large variations. 

Taking into account expression for $\delta=\alpha - d c^2 $ and $\alpha=\mu^{(b)} -\beta + \frac{b}{a}\Delta\mu$ (see (\ref{alpha})) one can easily verify that 
\begin{equation*}
\frac{\partial\hat\theta}{\partial \mu^{(s)}} (t,y)= 
\frac{1}{a d} 
\left[1+ \frac{\E}{y} \omega(T-t)
\right] \hbox{where}\ \ \omega(s) = \frac{1-e^{-\delta s}}{\delta}
\left(1-\frac{\Delta\mu}{\delta}\frac{b}{a} \right)
+ \Delta\mu \frac{b}{a} \frac{s\, e^{-\delta s}}{\delta}.
\end{equation*}
Again, as $\omega(0)=0$ and $\omega^\prime(s) = (1- \Delta\mu \frac{b}{a} s) \exp(-\delta s)$ we may conclude that $\omega(T-t)>0$, and, consequently 
$\frac{\partial\hat\theta}{\partial \mu^{(s)}} (t,y) >0$ provided that $1-\Delta\mu \frac{b}{a} T  \ge 0$. The latter condition can be reformulated as 
\[
\mu^{(s)} - \mu^{(b)} < 
\frac{[\sigma^{(s)}]^2 + [\sigma^{(b)}]^2 - 2 \varrho \sigma^{(s)} \sigma^{(b)}}
{\sigma^{(b)} (\sigma^{(b)} -  \varrho \sigma^{(s)}) T }.
\]
Also in this case the above structural condition is fulfilled because it reads as $\mu^{(s)} - \mu^{(b)} < 3.19$ and $\mu^{(s)} - \mu^{(b)} \approx 0.0512$ in the application discussed in the next section. Hence we may conclude that the optimal stock to bond proportion $\hat\theta$ is an increasing function with respect to the average stock return $\mu^{(s)}$.

\subsubsection{Sensitivity with respect to the growth rate}
Finally we discuss the dependence of the optimal value $\hat\theta$ on the growth rate $\beta$. 
Again, taking into account expression for $\delta=\alpha - d c^2 $ and $\alpha=\mu^{(b)} -\beta + \frac{b}{a}\Delta\mu$ we obtain  
\begin{equation*}
\frac{\partial\hat\theta}{\partial \beta } (t,y)= 
\frac{\Delta\mu}{a d} 
\frac{\E}{y} \omega(T-t)
\hbox{where}\ \ \omega(s) = \frac{1-e^{-\delta s}}{\delta^2}
-  \frac{s\, e^{-\delta s}}{\delta}.
\end{equation*}
Similarly as in the previous cases, as $\omega(0)=0$ and $\omega^\prime(s) =  s \exp(-\delta s)>0$, we may conclude that $\omega(T-t)>0$, and, consequently 
$\frac{\partial\hat\theta}{\partial \beta}(t,y) >0$.
Therefore the optimal stock to bond proportion $\hat\theta$ is an increasing function with respect to the wage growth $\beta$.

\section{Slovak pension system and calibration of the model parameters}

We have tested the proposed model on the second pillar of the Slovak pension system. According to Slovak legislature the percentage of salary transferred each year to a pension fund is 9\%, i.e. $\E=0.09$. It means that $\E$ can be considered as a small parameter. We have assumed the overall time period $T=40$ of saving of an individual pensioner. The average value of the wage growth in Slovakia for the period of 40 years has been adopted  from the paper by Kvetan et al. \cite{KVETAN2007} and has been estimated (in average value) as for 5\% p.a., i.e. $\beta=0.05$. Similarly as in Kilianov\' a et al. \cite{KILIANOVA2006}, stocks have been represented by the S\&P500 Index.  For the purpose of the comparison of results we have taken the same time period (Jan 1996-June 2002) yielding the average stock return $\mu^{(s)} = 10.28\%$ with the standard deviation $\sigma^{(s)} = 16.90\%$. As the modeling of bond returns is concerned we have considered the term structure of the zero coupon BRIBOR.\footnote{BRIBOR (Bratislava Interbank Offered Rate) is the former term structure in Slovakia till 1.1.2009}
Parameters of bond returns $\mu^{(b)}$ and their volatilities $\sigma^{(b)}$ have been taken from \cite{KILIANOVA2006} (see also calibration results of BRIBOR term structures from \cite{SC2}). We considered the average yield $\mu^{(b)}=5.16\%$ on the one year bond  with standard deviation $\sigma^{(b)}=0.882\%$ p.a. The correlation between stock and bond returns was set to $\varrho = -0.1151$. It is the same correlation values as in \cite{KILIANOVA2006}.

\begin{figure}
\begin{center}
\includegraphics[width=0.35\textwidth]{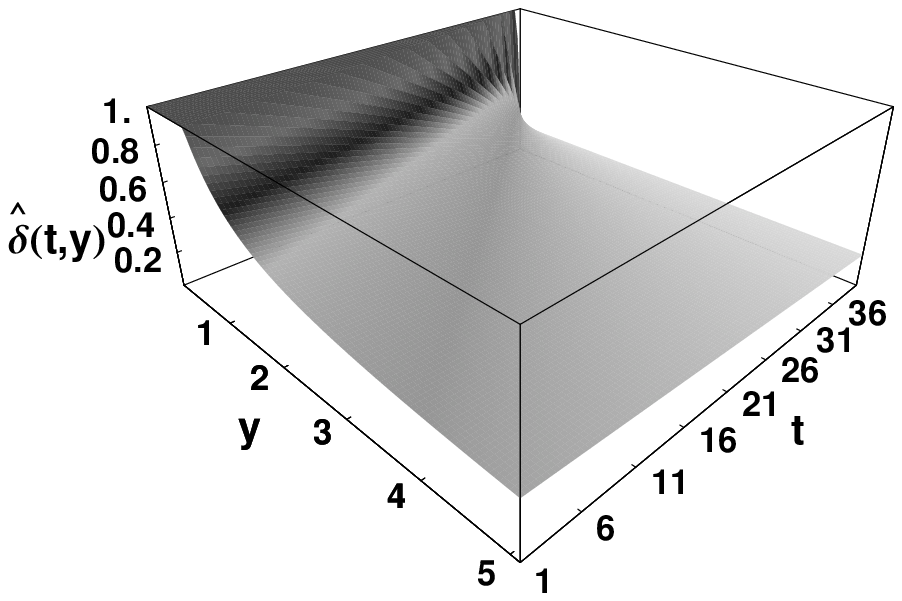}
\qquad
\includegraphics[width=0.25\textwidth]{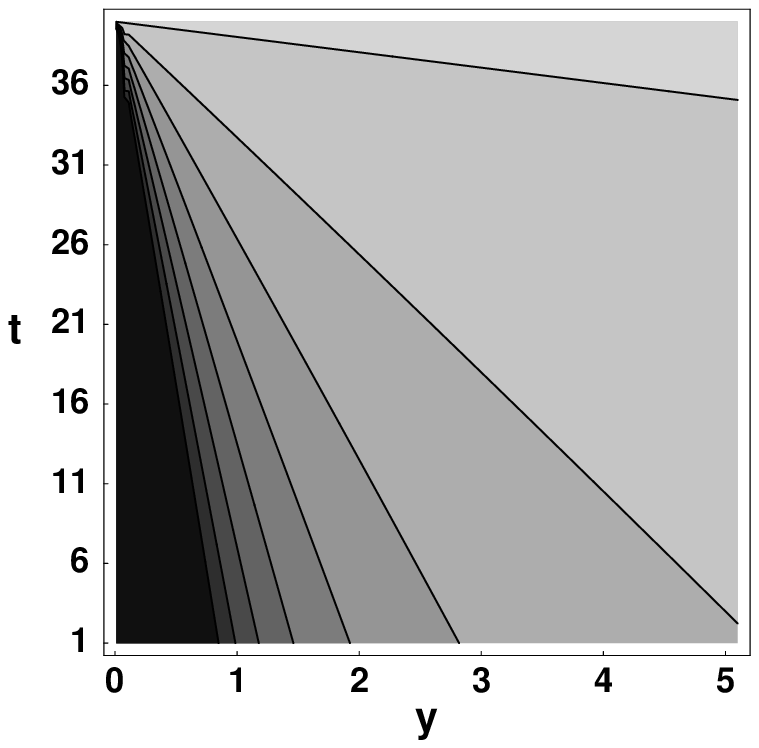}
\end{center}
\caption{
3D graph (left) and countour plot (right) of the optimal value 
$\hat\theta(t,y)$ computed from  the first order approximation of $\psi$.}
\label{obr-graf-analytic} 
\end{figure} 

In Fig.~\ref{obr-graf-analytic} we present the 3D plot as well as the contour plot of the optimal stock to bond proportion $\hat\theta(t,y)$ as a function of the time $t\in [0,T]$ and the level $y>0$ of saved yearly salaries. As for the approximate value of $\hat\theta(t,y)$  we considered the first order expansion given by (\ref{oThetaLin1}). We assumed the saver's risk aversion coefficient $d=10$. We cut-off values of the function $\hat\theta$ by the upper bound $1$, i.e. we in fact plotted the cutted function $(t,y) \mapsto \min\{\hat\theta(t,y), 1\}$.

\begin{figure}
\begin{center}
\includegraphics[width=0.35\textwidth]{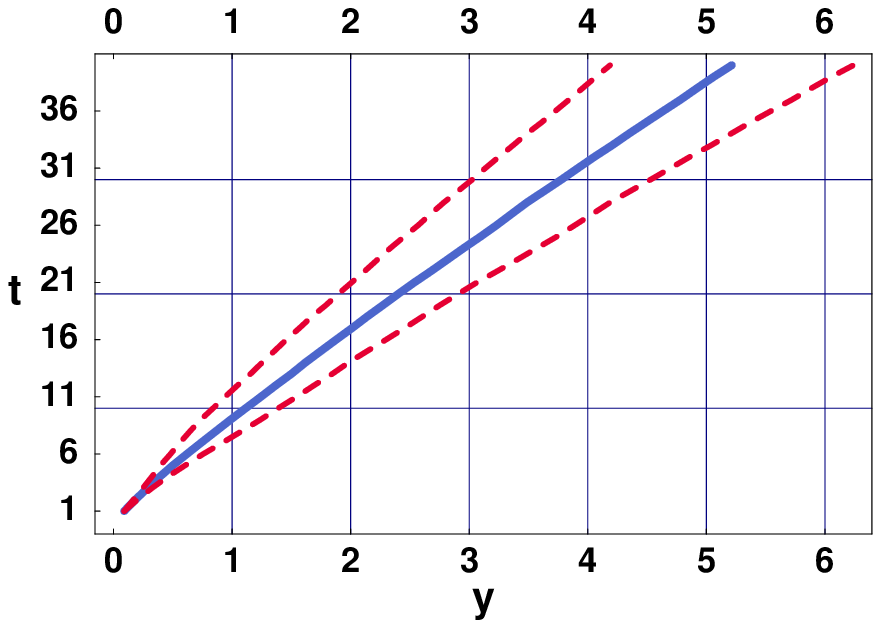}
\includegraphics[width=0.35\textwidth]{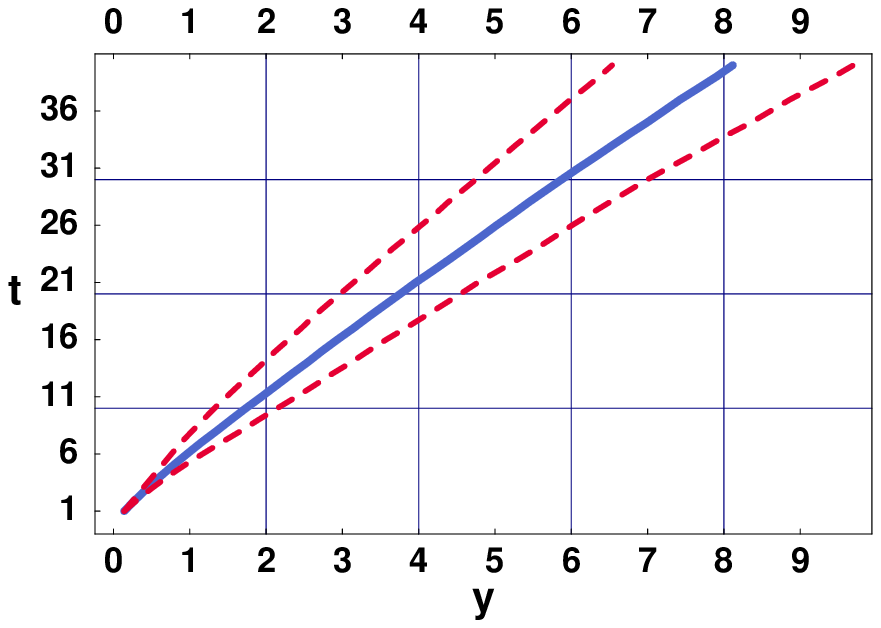}
\end{center}
\centerline{a) \ \ $\varepsilon=0.09$ \hglue2.5truecm b)\ \  $\varepsilon=0.14$}
\caption{
The mean wealth $\mathbb{E}(y_t)$ (solid line) and intervals $\mathbb{E}(y_t)\pm \sigma(y_t)$ (dashed lines) obtained by 10~000 Monte-Carlo simulation of the recurrent equation (\ref{CRE}) for different values of $\varepsilon$.}
\label{obr-graf-simul} 
\end{figure} 

In Fig.~\ref{obr-graf-simul} we present the mean wealth $\mathbb{E}(y_t)$ (solid line) obtained by 10~000 Monte-Carlo simulations of random paths $\{y_t,\  t=1,...,T\}$ calculated according to the recurrent equation 
$
y_{t+\tau} = F_t^{\tau}(\theta, y_t, Z), 
$
where $Z\sim N(0,1)$ and $\theta=\hat\theta(t,y_t)$. We take one year period $\tau=1$. The function $F_t^1$ was defined as in (\ref{CRE}), i.e. 
$F_t^{\tau}(\theta, y_t, z) = y_t \exp\left(  \left( \mu_t(\theta) - \beta_t -\frac{1}{2} \sigma_t^2(\theta)  \right) \tau  + \sigma_t(\theta) z \sqrt{\tau} \right) + \varepsilon \tau$. The dashed line represent the mean wealth plus/minus one standard deviation of the random variable $y_t$ at time $t\in[0,T]$. The simulation were obtained by using the optimal stock to bond proportion  $\theta=\hat\theta(t,y_t)$ depending on the value of simulated yearly saved salary $y_t$ at time $t$. For the parameter $\varepsilon=0.09$ (Slovakian pension system) we observe that at the end of simulation period $t=T$ the averaged saved salary $\mathbb{E}(y_T)\approx 5.2$ meaning that the saver following the optimal strategy given by $\hat\theta(t,y)$ has accumulated $5.2$ multiples of his last yearly salary. On the other hand, for the higher value  $\varepsilon=0.14$ (Bulgarian pension system) we obtain the averaged saved salary $\mathbb{E}(y_T)\approx 8.1$. 

\begin{figure}
\begin{center}
\includegraphics[width=0.35\textwidth]{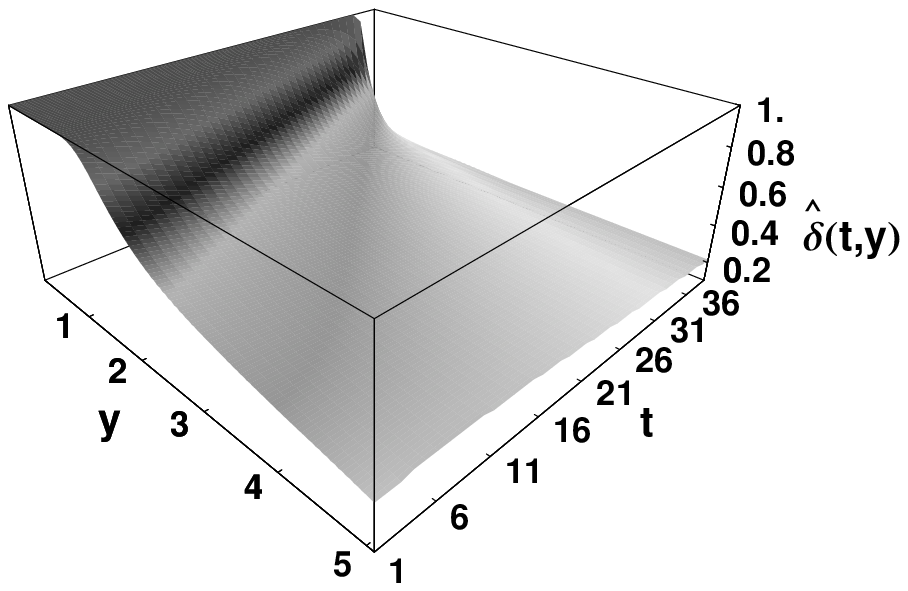}
\qquad
\includegraphics[width=0.25\textwidth]{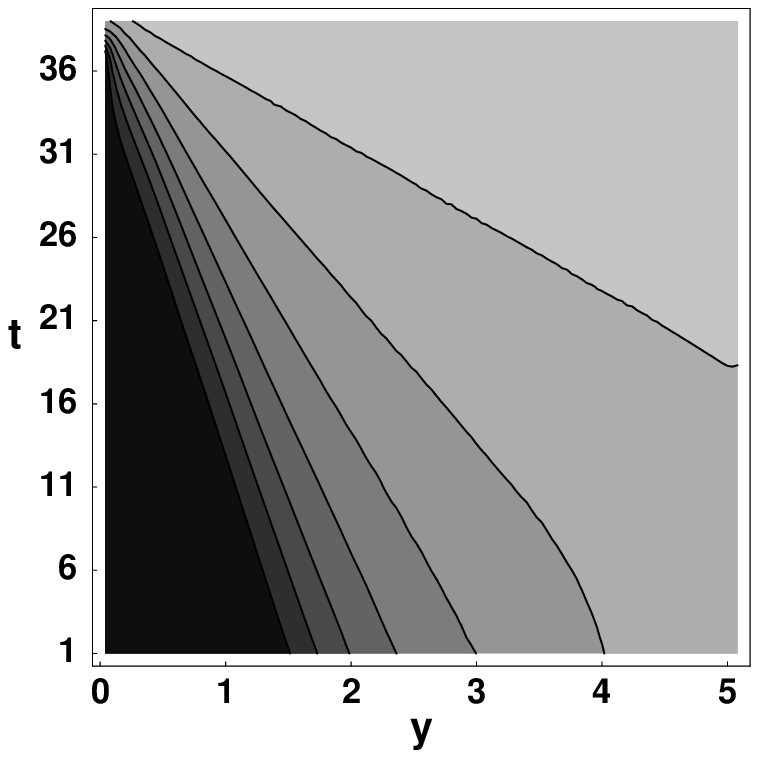}
\end{center}
\caption{
3D graph (left) and countour plot (right) of the optimal value $\hat\theta(t,y)$ computed by the discrete method from \cite{KILIANOVA2006}.}
\label{obr-graf-discrete} 
\end{figure} 

Finally, in Fig.~\ref{obr-graf-discrete} we present results that were computed for the same model parameters and saver's risk aversion $d=10$ but using the discrete model derived in \cite{KILIANOVA2006}. Notice that in the discrete model \cite{KILIANOVA2006} we restrict the optimal stock to bond proportion $\hat\theta$ to belong to the interval $\Delta_t\equiv [0,1]$. We can see and graphically compare the results obtained by our analytic approximation solution and those of discrete model have the same qualitative behavior. Comparison of the function $\hat\theta(t,y)$ computed by the time discrete and continuous dynamic stochastic model is depicted in Fig.~\ref{obr-graf-comparison}. The maximal difference $0.33$ is attained at $t=1$ and $y\approx1.3$. 

\begin{figure}
\begin{center}
\includegraphics[width=0.32\textwidth]{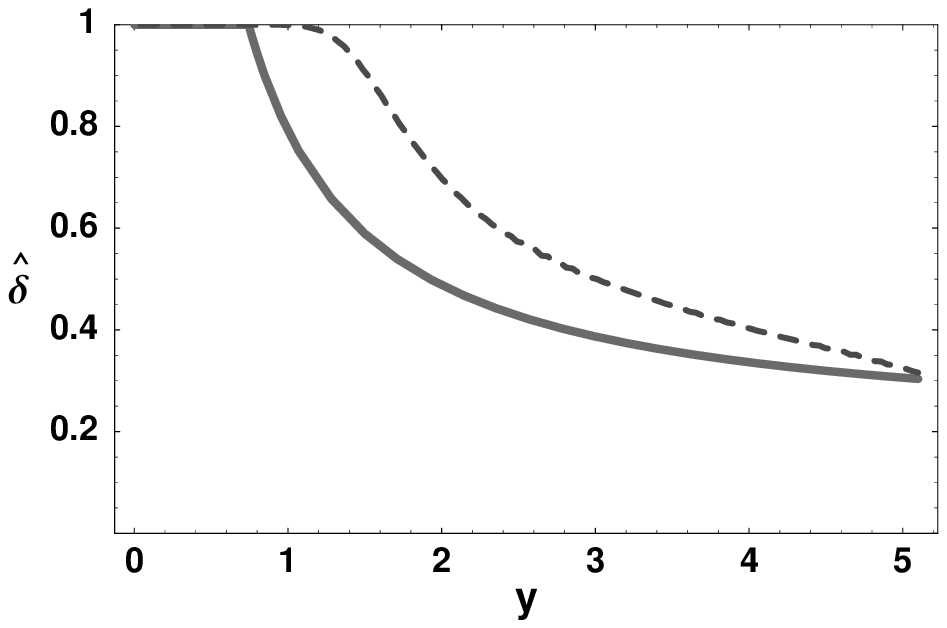}
\includegraphics[width=0.32\textwidth]{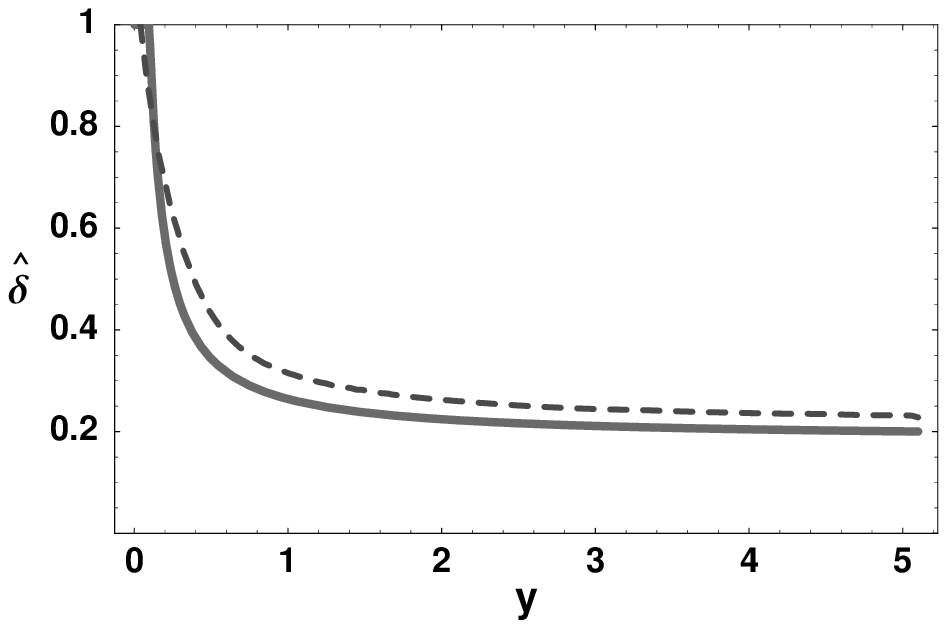}
\includegraphics[width=0.32\textwidth]{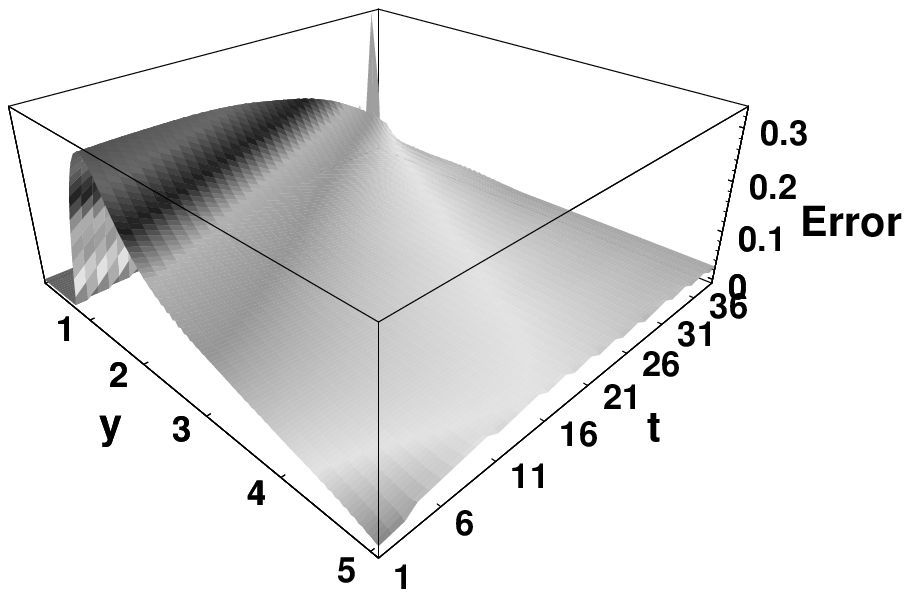}
\end{center}
\caption{
Comparison of the function $\hat\theta(t,y)$ computed by the discrete (dashed line) and continuous (solid line) model for different times a) $t=1$ and b) $t=35$. The overall error plot of their difference is shown in c).}
\label{obr-graf-comparison} 
\end{figure}

Recently, the dynamic stochastic accumulation model derived in Section 2 has been adopted and numerically tested for a Bulgarian pension saving system by Vulkov and Koleva (see \cite{VULKOV09}). The main difference between Slovakian and Bulgarian system consists in different values of the parameter $\varepsilon$. In the Slovak funded  pension system $\varepsilon=0.09$ whereas $\varepsilon=0.14$ for the Bulgarian funded pension pillar (see \cite{VULKOV09}). The impact of a higher value of the parameter $\varepsilon$ on the averaged saved sum ${\mathbb E}(d_T)$ is shown in Fig.~\ref{obr-graf-simul}.  Furthermore, they proposed an efficient numerical approximation scheme for solving the fully nonlinear parabolic equation (\ref{HJBE3e}). They showed the second order of convergence of their numerical scheme. The numerical method can handle general class of initial conditions (\ref{HJBE3c}) leading thus to possibility of considering a wider family of utility functions $U$ than CRRA class of functions.

\section{Conclusions}
We have analyzed a dynamic stochastic accumulation model for determining the optimal value of the stock to bond proportion in the pension saving decision. We showed how the problem can be formulated in terms of a solution to a fully nonlinear parabolic equation. We provided useful bounds on a solution to this equation. Moreover, we performed sensitivity analysis as well as numerical simulations of the model. By expanding a solution to the corresponding Hamilton--Jacobi--Bellman equation into power series we obtained a useful first order approximation that can be used in qualitative analysis of dependence of the optimal strategy on various model parameters. 

\section*{Acknowledgments}
The authors thank anonymous referees for their valuable comments and suggestion.
This research was supported by VEGA 1/0381/09 (CESIUK project) and bilateral Slovak--Bulgarian  project APVV SK-BG-0034-08.


\begin{thebibliography}{99}

\small

\bibitem{ISHIMURA2008}
R. Abe and  N. Ishimura,
Existence of Solutions for the Nonlinear Partial Differential Equation
arising in the Optimal Investment Problem,
Proc. Japan Acad.
{\bf 84}, Ser. A 
(2008),
11--14.


\bibitem{BODIE1992}
Z. Bodie, R. Merton and  W.F. Samuelson, 
Labor Supply Flexibility and Portfolio Choice in a Life-Cycle Model,
Journal of Economic Dynamics \& Control
{\bf 16} 
(1992),
427--449.


\bibitem{BODIE1995}
Z. Bodie,
On the Risk of Stocks in the Long Run,
Financial Analysts Journal
{\bf 51} 
(1995),
18--22.


\bibitem{BODIE2003}
Z. Bodie, J.B. Detemple, S. Otruba and S. Walter, 
Optimal Consumption--Portfolio Choices and Retirement Planning,
Journal of Economic Dynamics \& Control
{\bf 28}
(2003),
1115--1148.

\bibitem{BROWNE1995}
S. Browne,
Optimal Investment Policies for a Firm with a Random Risk Process: Exponential Utility and Minimizing the Probability of Ruin,
Math. Operations Research
{\bf 20}(4) 
(1995),
937--958.



\bibitem{FRIEND1975}
I. Friend and M.E. Blume,
The Demand for Risky Assets, 
The American Economic Review
{\bf 65}
(1975),
900--922.

\bibitem{ISHIMURA2009}
N. Ishimura and Y. Mita, 
A Note on the Optimal Portfolio Problem in Discrete Processes, 
Kybernetika
{\bf 45}(4)
(2009)

\bibitem{KILIANOVA2009}
S. Kilianov\'a,
Stochastic Dynamic Optimization Model for Pension Planning.
{\em Thesis, Comenius University} 2009, 106 pp.


\bibitem{KILIANOVA2006}
S. Kilianov\' a, I. Melicher\v c\'{\i}k and D.  \v{S}ev\v{c}ovi\v{c},
Dynamic Accumulation Model for the Second Pillar of the Slovak
Pension System, 
Czech Journal for Economics and Finance
{\bf 11-12}
(2006),
506--521.

\bibitem{KILIANOVA2009b}
S. Kilianov\'a and G. Pflug,
Optimal pension fund management under multi-period risk minimization,
Annals of Operations Research
{\bf 166} (1)
(2009),
261 -- 270.

\bibitem{VULKOV09}
M. Koleva and L. G. Vulkov,
Quasilinearization Numerical Scheme for Fully Nonlinear Parabolic Problems with Applications in Models of Mathematical Finance, 
submitted.

\bibitem{KVETAN2007}
V. Kvetan, M. Ml\'ynek, V. P\' alen\'{\i}k and  M. Radvansk\'y,
Starnutie, zdravotn\' y stav a determinanty v\'ydavkov na zdravie v podmienkach Slovenska, 
Research studies of Institute of Economics SAV Bratislava, 2007 (in Slovak), 
ISBN: 978-80-7144-160-1.


\bibitem{KWOK1998}
Y.K. Kwok,
{\em Mathematical Models of Financial Derivatives}.
New York, Heidelberg, Berlin:  Springer-Verlag,
1998.


\bibitem{MEHRA1985}
R. Mehra and E. Prescott,
The Equity Premium:  a Puzzle,
Journal of Monetary Economics
{\bf 15}
(1985),
145--161.

\bibitem{MS2009}
I. Melicher\v c\'{\i}k and  D. \v{S}ev\v{c}ovi\v{c},
Dynamic Stochastic Accumulation Model with Application to Pension Savings Management, 
submitted. 

\bibitem{MU2004}
I. Melicher\v c\'{\i}k and C. Ungvarsk\'y,
Pension Reform in Slovakia:  Perspectives of the Fiscal Debt and Pension Level, 
Czech Journal for Economics and Finance
{\bf 9-10}
(2004),
391--404.

\bibitem{MERTON1969}
R.C. Merton,
Lifetime Portfolio Selection Under Uncertainty:  The Continuous-Time Case,
Review of Economics and Statistics
{\bf 51}
(1969),
247--257.

\bibitem{MERTON1971}
R.C Merton,
Optimum Consumption and Portfolio Rules in a Continuous-Time Model, 
Journal of Economic Theory
{\bf 3}
(1971),
373--413.

\bibitem{MERTON1974}
R.C. Merton and  P.A. Samuelson,
Fallacy of the Log-Normal Approximation to Optimal Portfolio Decision-Making over many Periods,
Journal of Financial Economics
{\bf 1}
(1974),
67--94.

\bibitem{MERTON1975}
R.C. Merton,
Theory of Finance from the Perspective of Continuous Time,
Journal of Financial and Quantitative Analysis
{\bf 10}
(1975),
659--674.

\bibitem{MERTON1992}
R.C. Merton,
Continuous Time Finance, Revised ed., Blackwell
Publishers, Oxford, 1992.



\bibitem{PRATT1964}
J.W. Pratt,
Risk Aversion in the Small and in the Large, 
Econometrica
{\bf 32}
(1964),
122--136.


\bibitem{PROTTER} 
M. Protter and H.F. Weinberger,
Maximum Principles in  Differential Equations,
New York, Heidelberg, Berlin: Springer-Verlag, 1984.



\bibitem{SC2}
D. \v{S}ev\v{c}ovi\v{c} and A. Urb\'anov\'a Csajkov\'a,
On a Two-phase  Minmax Method for Parameter Estimation of the Cox, Ingersoll, and Ross Interest Rate Model,
Central European Journal of Operational Research
{\bf 13}
(2005),
169--188.


\bibitem{SAM1969}
P.A. Samuelson,
Lifetime Portfolio Selection by Dynamic Stochastic
Programming,
The Review of Economics and Statistics
{\bf 51}
(1969),
239--246.



\bibitem{YOUNG1990}
H.P. Young,
Progressive Taxation and Equal Sacrifice,
The American Economic Review
{\bf 80}
(1990),
253--266.



\end{thebibliography}
\end{document}